\documentclass{article}
\pdfoutput=1
\usepackage{amsthm,graphicx,fullpage}

\author{
  Jean-Guillaume~Dumas\thanks{%
    Universit\'e Grenoble Alpes,
    CNRS, LJK, 700 av. centrale, IMAG - CS 40700, 38058 Grenoble
    cedex 9, France, 
    \href{mailto:Jean-Guillaume.Dumas@imag.fr}{Jean-Guillaume.Dumas@imag.fr}%
    .}
  \and
  Vincent~Zucca\thanks{%
    Sorbonne Universit\'es, 
    Univ. Pierre et Marie Curie,
    Laboratoire LIP6, CNRS, umr 7606,
    4 place Jussieu, F75252 Paris, France, 
    \href{mailto:Vincent.Zucca@lip6.fr}{Vincent.Zucca@lip6.fr}%
    .}   
}

\newcommand{\sigsmall}{}
\newcommand{\sighspace}[1]{}
\newcommand{\sigvspace}[1]{}

\newcommand{\F}{\ensuremath{\mathbb{F}}}

\let\oldparagraph\paragraph
\renewcommand{\paragraph}[1]{\oldparagraph{#1.}}

\newtheorem{theorem}{Theorem}
\newtheorem{definition}[theorem]{Definition}
\newtheorem{lemma}[theorem]{Lemma}

\newtheorem{corollary}[theorem]{Corollary}
\newtheorem{remark}[theorem]{Remark}

\newcommand{\OpenDreamKit}{the \href{http://opendreamkit.org}{OpenDreamKit} \href{https://ec.europa.eu/programmes/horizon2020/}{Horizon 2020} \href{https://ec.europa.eu/programmes/horizon2020/en/h2020-section/european-research-infrastructures-including-e-infrastructures}{European Research Infrastructures} project (\#\href{http://cordis.europa.eu/project/rcn/198334_en.html}{676541})}
\title{Prover efficient public verification of dense or sparse/structured
  matrix-vector multiplication\thanks{This work is partly funded by
    \OpenDreamKit.}}

\usepackage{amsfonts,amsmath,calc,lipsum,enumerate}
\usepackage[all]{xy}
\usepackage{multirow}
\usepackage{paralist} 
\usepackage{refcount} 

\newcommand{\checks}{\ensuremath{\mathrel{\stackrel{?}{=\!=}}}}
\newcommand{\Z}{\ensuremath{\mathbb{Z}}}
\newcommand{\G}{\ensuremath{\mathbb{G}}}
\newcommand{\bigO}[1]{\ensuremath{{\mathcal O}\left(#1\right)}}
\newcommand{\bigTheta}[1]{\ensuremath{{\Theta}\left(#1\right)}}

\usepackage{todonotes}

\usepackage{svgcolor}
\usepackage[plainpages=true]{hyperref}
\makeatletter
\hypersetup{
pdftitle={Prover efficient public verification of dense or sparse/structured
  matrix-vector multiplication},
pdfauthor={J-G. Dumas, V. Zucca},
breaklinks=true,
colorlinks=true,
 linkcolor=darkred,
 citecolor=blue,
 urlcolor=darkgreen,
}
\makeatother

\begin{document}

\onecolumn\maketitle
\begin{abstract}
With the emergence of cloud computing services, computationally weak devices
(Clients) can delegate expensive tasks to more powerful entities (Servers). This
raises the question of verifying a result at a lower cost than that
of recomputing it. This verification can be private, between the Client and the
Server, or public, when the result can be verified by any third party. 
We here present protocols for the verification of matrix-vector
multiplications, that are secure against malicious Servers. 
The obtained algorithms are essentially optimal in the amortized model: the
overhead for the Server is limited to a very small constant factor, even in the
sparse or structured matrix case;
and the computational time for the public Verifier is linear in the dimension.
Our protocols combine probabilistic checks and cryptographic operations, but
minimize the latter to preserve practical efficiency. Therefore our
protocols are overall more than two orders of magnitude faster than
existing ones.
\end{abstract}

\section{Introduction}
With the emergence of cloud computing services, computationally weak devices
(Clients, such as smart phones or tablets) can delegate expensive tasks to more
powerful entities (Servers). 
Such heavy tasks can, e.g., be cryptographic operations, image manipulation or
statistical analysis of large data-sets.
This raises the question of verifying a result at a lower cost than that of
recomputing it. This verification can be private, between the Client and the
Server, or public, when the result can be verified by any third party.  

For instance within computer graphics (image compression and geometric
transformation), graph theory (studying properties of large networks), big data
analysis, one deals with linear transformations of large amount of data, often
arranged in large matrices with large dimensions that are in the order of
thousands or millions in some applications. 
Since a linear transformation on a vector $x$ can be
expressed by a matrix-vector multiplication (with a matrix of size
$m{\times}n$), a weak client can use one of the protocols in the literature
~\cite{Fiore:2012:PVD,Blanton:2014:matrix,Elkhiyaoui:2016:ETPV} to outsource and
verify this computation in the optimal time $O(m+n)$, i.e., linear in the input
and the output size.
However as these protocols use expensive cryptographic operations, such as
pairings, the constants hidden in the asymptotic complexity are usually
extremely large~\cite{Walfish:2015:VCWRT}.

In this paper, we propose an alternative protocol, achieving the same optimal
behavior, but which is also practical: the overhead for the Prover is now
very close to the time required to compute the matrix-vector
multiplication, thus gaining two orders of magnitude with respect to the
literature.
Our protocol not only does this for dense matrices, but is also sensitive to any
structure or sparsity of the linear transformation. 
For this, we first remove any quadratic operation that is not a
matrix-vector multiplication (that is we use projections and rank-1 updates) and
second we separate operations in the base field from cryptographic operations
so as to minimize the latter. 

More precisely, we first combine rank-one updates of~\cite{Fiore:2012:PVD} and
the projecting idea of~\cite{Elkhiyaoui:2016:ETPV} with Freivalds' probabilistic
check~\cite{Freivalds:1979:certif}. Second, we use a novel strategy of
vectorization. 
For instance, with a security parameter $s$ (e.g., an
$s=128$-bits equivalent security), exponentiations or pairings operations
usually cost about $O(s^3)$ arithmetic operations.
To make the whole protocol work practical, we thus reduce
its cost from $\bigO{s^3mn}$ to $\bigO{\mu(A)+s^3(m+n^{4/3})}$, where 
$\mu(A)<2mn$ is the cost of one, potentially structured, matrix-vector
multiplication. We also similarly reduce the work of the Verifier.
This allows us to gain two orders of magnitude on the Prover's work and
therefore on the overall costs of outsourcing, 
while preserving and sometimes even improving the practical efficiency of the
Verifier.

Thus, after some background in Section~\ref{sec:background}, 
our first
improvement is given in a relaxed public verification setting in
Section~\ref{sec:intfg} via matrix projection and probabilistic checks. 
Our second improvement is given in Section~\ref{sec:bootstrap} where the
verification is bootstrapped efficiently by vectorization. 
We then show how to combine all improvements in
Section~\ref{sec:fully} in order to obtain a complete and provably secure protocol. 
Finally, we show in Section~\ref{sec:expes} that our novel protocol indeed
induces a global overhead factor lower than $3$ with respect to non verified
computations. This is gaining several orders of magnitude on the Prover side 
with respect to previously known protocols, while keeping the Verification step
an order of magnitude faster.

\section{Background and definitions}\label{sec:background}
In this paper, we want to be able to prove fast that a vector is a solution to a
linear system, or equivalently that a vector is the product of another vector by
a matrix. 
This is useful, e.g., to perform some statistical analysis on some medical data.
We distinguish the matrix, a static data, from the vectors which
are potentially diverse. In the following, $\F_p$ will denote a prime field and
we consider: 
\begin{compactitem}
\item Data: matrix $A\in\F_p^{m{\times}n}$.
\item Input: one or several vectors $x_i\in\F_p^n$, for $i=1..k$.
\item Output: one or several vectors $y_i=Ax_i\in\F_p^m$, for $i=1..k$.
\end{compactitem}
Then, we denote by $\star{}$ an operation performed in the exponents (for
instance, for $u\in\G^n$ and $v\in\Z^n$, the operation $u^T\star{}v$ actually denotes
$\prod_{j=1}^n u[j]^{v[j]}$).

\paragraph{Publicly Verifiable Computation}
A publicly verifiable computation scheme, in the formal setting
of~\cite{Parno:2012:delegate}, is in fact four algorithms 
(\emph{KeyGen}, \emph{ProbGen}, \emph{Compute}, \emph{Verify}),
where \emph{KeyGen} is some (amortized) preparation of the data, \emph{ProbGen}
is the preparation of the input, \emph{Compute} is the work of the \emph{Prover}
and \emph{Verify} is the work of the \emph{Verifier}. 
Usually the Verifier also executes \emph{KeyGen} and \emph{ProbGen} but in a
more general setting these can be performed by different entities (respectively
called a \emph{Preparator} and a \emph{Trustee}). More formally we define these
algorithms as follow:

\begin{compactitem}
\item \emph{KeyGen}$(1^\lambda,f)\rightarrow(\texttt{param}, EK_f, VK_f)$: a randomized algorithm
  run by a \emph{Preparator}, it takes as input a security parameters $1^\lambda$ and the function~$f$ to be outsourced. It outputs public parameters $\texttt{param}$ which will be used by the three remaining algorithms, an evaluation key $EK_f$ and a verification key $VK_f$.
\item \emph{ProbGen}$(x)\rightarrow(\sigma_x)$: a randomized algorithm run by a \emph{Trustee} which takes as input an element $x$ in the domain of the outsourced function $f$. It returns $\sigma_x$, an encoded version of the input $x$.
\item \emph{Compute}$(\sigma_x,EK_f)\rightarrow(\sigma_y)$: an algorithm run by the \emph{Prover} to compute an encoded version $\sigma_y$ of the output $y=f(x)$ given the encoded input $\sigma_x$ and the evaluation key $EK_x$.
\item \emph{Verify}$(\sigma_y,VK_f)\rightarrow y \text{ or} \perp$: given the encoded output $\sigma_y$ and the verification key $VK_f$, the $\emph{Verifier}$ runs this algorithm to determine whether $y=f(x)$ or not. If the verification passes it returns $y$ otherwise it returns an error $\perp$.   \end{compactitem}

	\paragraph{Completeness} A publicly verifiable computation scheme for a family of function $\mathcal{F}$ is considered to be \textit{perfectly complete} (or \textit{correct}) if for every function belonging to $\mathcal{F}$ and for every input in the function domain, an honest \emph{Prover} which runs faithfully the algorithm $\emph{Compute}$ will \textit{always} (with probability 1) output an encoding $\sigma_y$ which will pass \emph{Verify}.

\paragraph{Soundness} A publicly verifiable computation scheme for a family of function $\mathcal{F}$ is called \textit{sound} when a prover cannot convince a verifier to accept a wrong result $y'\neq y$ except with negligible probability. More formally we evaluate the capability of an adversary $\mathcal{A}$ to deceive the verifier through a \textit{soundness experiment}.
In this experiment, we assume that the adversary $\mathcal{A}$ accesses to the output of the algorithm $\emph{KeyGen}$ by calling an oracle $\mathcal{O}_{\text{\emph{KeyGen}}}$ with inputs $1^\lambda$ and the function to evaluate $f$. This oracle $\mathcal{O}_{\text{\emph{KeyGen}}}$ returns public parameters for the protocol $\texttt{param}$, an evaluation key $EK_f$ and a verification key $VK_f$. Afterwards the adversary $\mathcal{A}$ sends its challenge input $x$ to an oracle $\mathcal{O}_{\text{\emph{ProbGen}}}$ which returns $\sigma_{x}$. Finally $\mathcal{A}$ outputs an encoding $\sigma_{y^*}\neq \sigma_y$ and runs the \emph{Verify} algorithm on inputs $\sigma_{y^*}$ and $VK_{f}$, whether it outputs $y$ or $\perp$ the experiment has either succeeded or failed.

\begin{definition}
  A publicly verifiable computation scheme for a family of function $\mathcal{F}$ is sound if and only if for any polynomially bounded adversary $\mathcal{A}$ and for any $f$ in $\mathcal{F}$ the probability that $\mathcal{A}$ succeeds in the soundness experiment is negligible in the security parameter.
\end{definition}

\paragraph{Adversary model}
The protocol in~\cite{Fiore:2012:PVD} (recalled for the sake of completeness in Appendix~\ref{sec:fg})
is secure against a \emph{malicious Server
only}. That is the Client must trust both the Preparator and the Trustee.
We will stick to this model of attacker in the remaining of this paper.
 Otherwise some attacks can be mounted:
 \begin{compactitem}
 \item Attack with a \emph{Malicious Preparator only}: send $A'$ and a correctly
   associated $W'$ to the Server, but pretend that $A$ is used. Then, all
   verifications do pass, but for $y'=A'x$ and not $y=Ax$.
 \item Attack with a \emph{Malicious Trustee only}: there, a malicious
   assistant can provide a wrong $\text{VK}_x$ making the verification fail even
   if the Server and Preparator are honest and correct. 
\item Attack with \emph{Malicious Server and Trustee}: 
  the Server sends any $y'$ and any $z'$ to the Trustee, who computes 
  $\text{VK'}_x[i]=e(z'[i];g_2)/a^{y'[i]}$, that will match the verification
  with $y'$ and $z'$.
 \end{compactitem}

\paragraph{Public delegatability} 
One can also further impose that there is not interaction between the Client and
the Trustee after the Client has sent his input to the Server. Publicly
verifiable protocols with this property are said to be publicly
delegatable~\cite{Elkhiyaoui:2016:ETPV}. The protocol in~\cite{Fiore:2012:PVD}
does not achieve this property, but some variants
in~\cite{Blanton:2014:matrix,Elkhiyaoui:2016:ETPV} already can. 

\paragraph{Bilinear Pairings}
The protocols we present in this paper use bilinear pairings and their security is based on the co-CDH assumption, for the sake of completeness we recall hereafter these definitions.  
\begin{definition}[\bf{bilinear pairing}]\ \\
  Let $\G_1$, $\G_2$ and $\G_T$ be three groups of prime order $p$, a bilinear
  pairing is a map $e:\G_1\times \G_2 \rightarrow \G_T $ with the following
  properties: 
  \begin{compactenum}
  \item {\em bilinearity}:
    $\forall{}a,b\in\F_p,~\forall{}(g_1,g_2)\in\G_1\times\G_2,~e(g_1^a,g_2^b)=e(g_1,g_2)^{ab}$; 
  \item {\em non-degeneracy}: if $g_1$ and $g_2$ are generators of $\G_1$ and
    $\G_2$ respectively then $e(g_1,g_2)$ is a generator of $\G_T$;
  \item {\em computability}: $\forall (g_1,g_2)\in\G_1\times\G_2$, there exist an efficient algorithm to compute
    $e(g_1,g_2)$.
  \end{compactenum}
\end{definition}

\begin{definition}[\bf{co-CDH assumption}]\ \\
  Let $\G_1$, $\G_2$ and $\G_T$ be three groups of prime order $p$, such that
  there exist a bilinear map $e:\G_1\times \G_2\rightarrow \G_T$. Let
  $g_1\in\G_1$, $g_2\in\G_2$ be generators and $a$,
  $b\overset{\$}{\leftarrow}\F_p$ be chosen randomly. We say that the
  co-computational Diffie-Hellman assumption (\mbox{co-CDH}) holds in $\G_1$, if
  given 
  $g_1$, $g_2$, $g_1^a$, $g_2^b$ the probability to compute $g_1^{ab}$ is
  negligible.
\end{definition}

\paragraph{Related work}
The work of~\cite{Fiore:2012:PVD} introduced the idea of performing twice
the computations, once in the classical setting and once on encrypted values. This
enables the Client to only have to check consistency of both results. The
protocol is sound under the Decision Linear and \mbox{co-CDH} hypothesis. 
Then~\cite{Blanton:2014:matrix} extended the part on matrix-vector
multiplication to matrix-matrix while adding public delegatability. 
Finally, \cite{Elkhiyaoui:2016:ETPV} introduced the idea of projecting the
random additional matrix and the extra-computations of the server, which allows to
reduce the cost of the \emph{Verify} algorithm. It also decreases the size of the
verification key by a factor $m$.
For an $m{\times}n$ dense matrix, the protocol in~\cite{Fiore:2012:PVD} has a
constant time overhead for the Prover, but this constant is on the order of
cryptographic public-key operations like pairings. Similarly, the Verifier has
$\bigO{mn}$ cryptographic public-key pre-computations and $\bigO{n}$ of these
for the public verification. These cryptographic operations can then induce some
$10^6$ slow-down~\cite{Walfish:2015:VCWRT} and do no improve even if the initial
matrix is sparse or structured (as the rank one updates, $s\cdot t^T$ and
$\sigma\cdot{}\tau^T$, are always dense).

\section{A first step towards public verifiability}\label{sec:intfg} 
Freivalds' probabilistic verification of matrix
multiplications~\cite{Freivalds:1979:certif} allows for private verifiability of
matrix-vector computations. This can be naturally extended in the random oracle
model via Fiat-Shamir heuristic~\cite{Fiat:1986:Shamir}.
This however forces the vectors to be multiplied to be known in advance (the full
details are given in Appendix~\ref{sec:probaverif}), whereas our goal  
is instead to obtain public verifiability with an \emph{unbounded number of
  vector inputs}.
As an upstart, we thus first present an improvement if the public verification
model is slightly relaxed: in this Section, we allow the Trustee to perform some
operations after the computations of the Server. We will see in next sections
how to remove the need for the Trustee's intervention.
For this, we combine Freivalds projection (to check that $Ax_i=y_i$, one can first
precompute $w^T=u^TA$ and check that $w^Tx_i=u^Ty_i$) with Fiore \&
Gen\-na\-ro's protocol, in order to improve the running time of both the Trustee
and the Client: we let the Prover compute its projection in the group. That way
most of the pairings computations of the Trustee and Client are transformed to
classical operations: the improvement is from $\bigO{n}$ cryptographic
operations to $\bigO{n}$ classical operations and a single cryptographic one. 
Further, the projection can be performed beforehand, during the precomputation
phase. That way the preparation requires only one matrix-vector for the
Freivalds projection and the dense part is reduced to a single vector. 
The cryptographic operations can still be delayed till the last check on
pairings. This is shown in Figure~\ref{fig:spmv}.

\begin{figure}[htbp]\centering\sigvspace{-5pt}
\fbox{\begin{minipage}{0.95\textwidth}
\begin{compactitem}
\item Preparator: secret random $u\in\F_p^m$, $t\in\F_p^n$, then
  $\omega^T=g_1^{u^TA+t^T}\in\G_1^n$. 
\item Preparator to Prover: $A\in\F_p^{m{\times}n}$, $\omega\in\G_1^n$
\item Preparator to Trustee: $u$, $t$ in a secure channel.
\item Verifier to Prover: $x_i\in\F_p^n$
\item Prover to Verifier: $y_i\in\F_p^m$, $\zeta_i\in\G_1$ such that $y_i=Ax_i$
  and $\zeta_i=\omega^T\star{}x_i$.
\item Verifier to Trustee: $x_i$, $y_i$
\item Trustee to Verifier:
  $h_i=(u^T\cdot{}y_i)\in\F_p$ and $d_i=(t^T\cdot{}x_i)\in\F_p$, then send
  $\eta_i=e(g_1;g_2)^{h_i+d_i}\in\G_T$.
\item Verifier public verification:
  $e(\zeta_i;g_2)\checks{}\eta_i$ in $\G_T$.
\end{compactitem}
\end{minipage}}
\caption{Interactive protocol for Sparse-matrix vector multiplication
  verification under the \mbox{co-CDH}.}\label{fig:spmv}
\sigvspace{-15pt}\end{figure}

\begin{theorem}\label{thm:sparsecdh} The protocol of Figure~\ref{fig:spmv} is
  perfectly complete and sound under the co-Computational Diffie-Hellman
  assumption.
\end{theorem}
The proof of Theorem~\ref{thm:sparsecdh} is given in
Appendix~\ref{ssec:proofsparse}.

\section{Verifying the dot-products by bootstrapping and
  vectorization}\label{sec:bootstrap}
To obtain public verifiability and public delegatability, the Client should
perform both dot-products, $u^T\cdot{}y$ and $t^T\cdot{}x$ (from now on, for the
sake of simplicity, we drop the indices on~$x$ and~$y$).
But as~$u$ and~$t$ must remain secret, they will be encrypted 
beforehand. To speed-up the Client computation, the idea is then to let the
Server perform the encrypted dot-products and to allow the Client to verify
them mostly with classical operations.

For this trade-off, we use vectorization. That is, for the vectors~$u$ and~$y$,
we form another representation as $\sqrt{m}{\times}\sqrt{m}$ matrices:  
\[U=\left[\begin{array}{ccc}
u_1&\ldots&u_{\sqrt{m}} \\
u_{1+\sqrt{m}}&\ldots&u_{2\sqrt{m}} \\
\ldots&\ldots&\ldots\\
u_{1+m-\sqrt{m}}&\ldots&u_{m} \\
\end{array}\right]
~~\text{and}~~
Y=\left[\begin{array}{ccc}
y_1&\ldots&y_{1+m-\sqrt{m}} \\
y_{2}&\ldots&y_{2+m-\sqrt{m}} \\
\ldots&\ldots&\ldots\\
y_{\sqrt{m}}&\ldots&y_{m} \\
\end{array}\right].\]
Then $u^T\cdot{}y=Trace(U Y)$. Computing with this representation is in general
slower than with the direct dot-product, $\bigO{\sqrt{m}^3}$ instead of
$\bigO{m}$. As shown next, this can be
circumvented with well-chosen left-hand sides and at least mitigated, with
unbalanced dimensions.

\subsection{dot-product with rank~1 left-hand side}

The first case is if~$u$ is of rank~$1$, that is if in matrix form, $u$~can be
represented by a rank one update, $U=\mu\cdot\eta^T$ for
$\mu,\eta\in\F_p^{\sqrt{m}}$. Then both representations require roughly the same
number of operations to perform a dot-product since then: 
\begin{equation}\label{eq:tracerk1}
Trace(\mu\cdot\eta^T\cdot{}Y)=\eta^T\cdot{}Y\cdot{}\mu
\end{equation}
Therefore, we let the Prover compute 
$z^T=g_1^{\eta^T}\star{}Y$, where $z[i]=g_1^{\sum \eta[j]Y[j,i]}$, and then the
Verifier can check this value via Freivalds with a random vector~$v$:
$g_1^{\eta^T}\star\left(Y\cdot{}v\right)\checks{}z^T\star{}v$.
The point is that the Verifier needs now $\bigO{m}$ operations to compute
$\left(Y\cdot{}v\right)$, but these are just classical operations over the
field. Then its 
remaining operations are cryptographic but there is only $\bigO{\sqrt{m}}$ of
these. 
Finally, the Verifier concludes the computation of the dotproduct, still with
cryptographic operations, but once again with only $\bigO{\sqrt{m}}$ of them.
Indeed, the dot product
$d=u^Ty=Trace(UY)=Trace(\mu\cdot\eta^T\cdot{}Y)=\eta^T\cdot{}Y\cdot{}\mu$
is checked by $e(g_1^d;g_2)=e(g_1;g_2)^{d}=\prod_{i=1}^{\sqrt{m}}
e(z[i];g_2^{\mu[i]})$ and the latter is
$e(g_1;g_2)^{\sum\sum\mu[i]\eta[j]Y[j,i]}=e(g_1;g_2)^{\eta^T\cdot{}Y\cdot{}\mu}$.

In practice, operations in a group can be slightly faster than pairings.
Moreover $\lceil\sqrt{m}\rceil^2$ can be quite far off $m$.
Therefore it might be interesting to use a non square vectorization
$b_1{\times}b_2$, as long as $b_1b_2\geq m$ and
$b_1+b_2=\bigTheta{\sqrt{m}}$ (and $0$~padding if needed). 
Then we have $U\in\F_p^{b_1{\times}b_2}$, $\mu\in\F_p^{b_1}$, $\eta\in\F_p^{b_2}$,
$Y\in\F_p^{b_2{\times}b_1}$ and $z\in\G_2^{b_1}$.
The obtained protocol can compute
$e(g_1;g_2)^{u^T\cdot{}y}$ with $\bigO{\sqrt{m}}$ cryptographic operations on
the Verifier side and is given in Figure~\ref{fig:rk1dp}.


\begin{lemma}\label{lem:rk1dp}
 The protocol of Figure~\ref{fig:rk1dp} for publicly delegation
 of a size $m$ external group dot-product verification with \mbox{rank-$1$} left hand
 side is sound, perfectly complete and requires the following number of
 operations where $b_1b_2\geq m$ and $b_1+b_2=\bigTheta{\sqrt{m}}$:
\begin{compactitem}
\item Preparation: $\bigO{b_1{+}b_2}$ in $\G_i$;
\item Prover: $\bigO{m}$ in $\G_i$;
\item Verifier: $\bigO{m}$ in $\F_p$, $\bigO{b_1{+}b_2}$ in $\G_i$ and $\bigO{b_1}$ pairings.
\end{compactitem}
\end{lemma}
\begin{proof} Correctness is ensured by Equation~(\ref{eq:tracerk1}). Soundness
  is 
  given by Freivalds check.
Complexity is as given in the Lemma: indeed, for the Verifier, we have: 
for $Y\cdot{}v$: $\bigO{b_2b_1}=\bigO{m}$ classic operations;
for $g_1^{\eta^T}\star(Yv)$: $\bigO{b_2}$ cryptographic (group) operations;
for $z^T\star{}v$: $\bigO{b_1}$ cryptographic (group) operations; 
and for $\displaystyle\prod_{i=1}^{b_1} e(z[i];g_2^{\mu[i]})$: $\bigO{b_1}$
  cryptographic (pairings) operations. 
Then the preparation requires to compute $g_1^\eta\in\G_1^{b_2}$ and $g_2^\mu\in\G_1^{b_1}$, 
while the Prover needs to compute $g_1^{\eta^T}\star{}Y$ for
$Y\in\F_p^{b_2{\times}b_1}$ and $b_1b_2=\bigO{m}$.
\end{proof}

\begin{figure}[!htb]\centering
\fbox{\begin{minipage}{0.95\textwidth}
  \begin{compactitem}
\item \emph{KeyGen}$(1^\lambda,\mu,\eta)$: given the security parameter $1^\lambda$ and vectors $\mu\in\mathbb{F}_p^{b_2}$ and $\eta\in\mathbb{F}_p^{b_1}$ such that $u = \mu\cdot\eta^t$, it selects two cyclic groups $\mathbb{G}_1$ and $\mathbb{G}_2$ of prime order $p$ that admit a bilinear pairing $e: \mathbb{G}_1\times  \mathbb{G}_2 \rightarrow  \mathbb{G}_T$ and generators $g_1$, $g_2$ and $g_T$ of the three groups. Finally it outputs $\texttt{params}=\{b_1,b_2,p,\mathbb{G}_1,\mathbb{G}_2,\mathbb{G}_T,e,g_1,g_2,g_T\}$ and $EK_f=(g_1^{\eta^T})$ and $VK_f=(g_2^\mu)$.
  
\item \emph{ProbGen}$(y)$: from $y\in\mathbb{F}_p^m$ it builds $Y\in\mathbb{F}_p^{b_1\times b_2}$ and outputs $\sigma_x=Y$.
  
\item \emph{Compute}$(\sigma_x,EK_f)$: compute $z^T=g_1^{\eta^T}\star Y$ and outputs $\sigma_y = (z^T)$. 
  
\item \emph{Verify}$(\sigma_y,VK_f)$: it starts by sampling randomly a vector $v\in\mathbb{F}_p^{b_2}$ then it checks whether $z^T\star v$ is equal to $g_1^{\eta^T}\star(Yv)$ or not. If the test passes it returns $\prod e(z[i];g_2^{\mu[i]})$ and if it fails it returns $\perp$. 
  \end{compactitem}
\end{minipage}}
  \caption{Publicly delegatable protocol for the dot-product in an external group
  with a rank-$1$ left hand side.}\label{fig:rk1dp}
\sigvspace{-15pt}\end{figure}

\subsection{Rectangular general dot-product}
Now if $u$ is not given by a rank $1$ update, one can still verify a dot-product
with only $\bigO{\sqrt{m}}$ pairings operations but as the price of slightly
more group operations as given in Figure~\ref{fig:gendp}.

\begin{figure}[!htb]\centering
\fbox{\begin{minipage}{0.95\textwidth}
  \begin{compactitem}
\item \emph{KeyGen}$(1^\lambda,u)$: given the security parameter $1^\lambda$ and vector $u\in\mathbb{F}_p^{m}$, it selects two cyclic groups $\mathbb{G}_1$ and $\mathbb{G}_2$ of prime order $p$ that admit a bilinear pairing $e: \mathbb{G}_1\times  \mathbb{G}_2 \rightarrow  \mathbb{G}_T$ and generators $g_1$, $g_2$ and $g_T$ of the three groups and also integers $b_1$, $b_2$ such that $m = b_1b_2$ and it outputs  $\texttt{params}=\{b_1,b_2,p,\mathbb{G}_1,\mathbb{G}_2,\mathbb{G}_T,e,g_1,g_2,g_T\}$. Then it samples a random $w\in\mathbb{F}_p^{b_1}$, creates $U\in\mathbb{F}_p^{b_1\times b_2}$ from $u$ and finally it outputs and $EK_f=(g_1^U)$ and $VK_f=(g_1^{w^T\cdot U},g_2^{\omega^T})$.
  
\item \emph{ProbGen}$(y)$: from $y\in\mathbb{F}_p^m$ it builds $Y\in\mathbb{F}_p^{b_2\times b_1}$ and outputs $\sigma_x=Y$.
  
\item \emph{Compute}$(\sigma_x,EK_f)$: compute $C=g_1^{U}\star Y$ and outputs $\sigma_y = (C)$. 
  
\item \emph{Verify}$(\sigma_y,VK_f)$: it starts by sampling randomly a vector $v\in\mathbb{F}_p^{b_1}$ then it computes $z=C\star v$ and checks whether $\prod e(z[i];g_2^{\omega[i]})$ is equal to $e(g_1^{w^T\cdot U} \star (Yv) ; g_2)$ or not. If the test passes it returns $Trace(C)$ and if it fails it returns $\perp$. 
  \end{compactitem}
\end{minipage}}
\caption{Publicly delegatable protocol for the external
  dot-product.}\label{fig:gendp}
\sigvspace{-5pt}\end{figure}


\begin{lemma}\label{lem:gendp}
 The protocol of Figure~\ref{fig:gendp} is sound, perfectly
  complete and requires the following number of operations with $b_1b_2\geq m$:
\begin{compactitem}
\item Preparation: $\bigO{m}$ in $\F_p$ and $\bigO{m}$ in $\G_i$;
\item Prover: $\bigO{mb_1}$ in $\G_1$;
\item Verifier: $\bigO{m}$ in $\F_p$, $\bigO{b_1^2{+}b_2}$ in $\G_i$ and
  $\bigO{b_1}$ pairings.
\end{compactitem}
\end{lemma}
\begin{proof} Correctness is ensured by the vectorization in
  Equation~(\ref{eq:tracerk1}).  
Soundness is given by the Freivalds check.
Complexity is as given in the Lemma: indeed, for the Verifier, we have: 
\begin{compactenum}
\item $Y\cdot{}v$: $\bigO{b_2b_1}=\bigO{m}$ classic operations;
\item $g_1^{w^TU}\star(Yv)$: $\bigO{b_2}$ cryptographic (group) operations;
\item $z=C\star{}v$: $\bigO{b_1^2}$ cryptographic (group) operations;
\item $\displaystyle\prod_{i=1}^{b_1} e(z[i];g_2^{w[i]})$: $\bigO{b_1}$
  cryptographic (pairings) operations; 
\end{compactenum}
Then the preparation requires to compute $w^T{\times}U$. This is
$\bigO{b_1b_2=m}$ operations. Finally, the Prover needs to compute the matrix
multiplication $g_1^{U}\star{}Y$ for $U\in\F_p^{b_1{\times}b_2}$ and
	$Y\in\F_p^{b_2{\times}b_1}$, in $\bigO{b_1^2b_2}=\bigO{mb_1}$.
\end{proof}

Therefore, one can take $b_1=\bigO{\sqrt[3]{m}}$ and $b_2=\bigO{m^{2/3}}$
which gives only $\bigO{m^{2/3}}$ cryptographic operations for the Verifier,
and $\bigO{m^{4/3}}$  cryptographic operations for the Prover.
Now a dot-product can be cut in $\frac{n}{k}$ chunks of size $k$. 
Then Each chunk can be checked with the protocol of Figure~\ref{fig:gendp} and
the final dot-product obtained by adding the chunks. This gives the complexity
of Corollary~\ref{cor:gendp}.

\begin{corollary}\label{cor:gendp}
There exist a protocol for the dot-product using:
$\bigO{n^{1-a/3}}$ cryptographic operations for the Verifier and 
$\bigO{n^{1+a/3}}$ cryptographic operations for the Prover, for any $0<a<1$.
\end{corollary}
\begin{proof}
We let $k=n^{a}$ and use $n/k$ times the protocol of the previous point
with $b_1=\sqrt[3]{k}$ and $b_2=k^{\frac{2}{3}}$. Overall this gives
$n/k(k^{4/3})$ for the Prover and $n/k(k^{2/3})$ for the Verifier.
\end{proof}

\section{Public delegatability via bootstrapping}\label{sec:fully} 
To recover the public delegatability model, we use the protocol of
Figure~\ref{fig:spmv} but we trade back some
cryptographic operations using the protocol of Figure~\ref{fig:rk1dp} to the
Verifier. With an initial
matrix $A\in\F_p^{m{\times}n}$ we however trade back only on the order of
$\bigO{\sqrt{m}+\sqrt{n}}$ cryptographic operations.
This gives a slower verification in practice but interaction
is not needed anymore.
We present our full novel protocol for matrix vector product in
Figure~\ref{fig:matrix-vector} (with the flow of exchanges shown in
Figure~\ref{fig:provprot}, Appendix~\ref{sec:protofig}).
\begin{figure}[!htb]
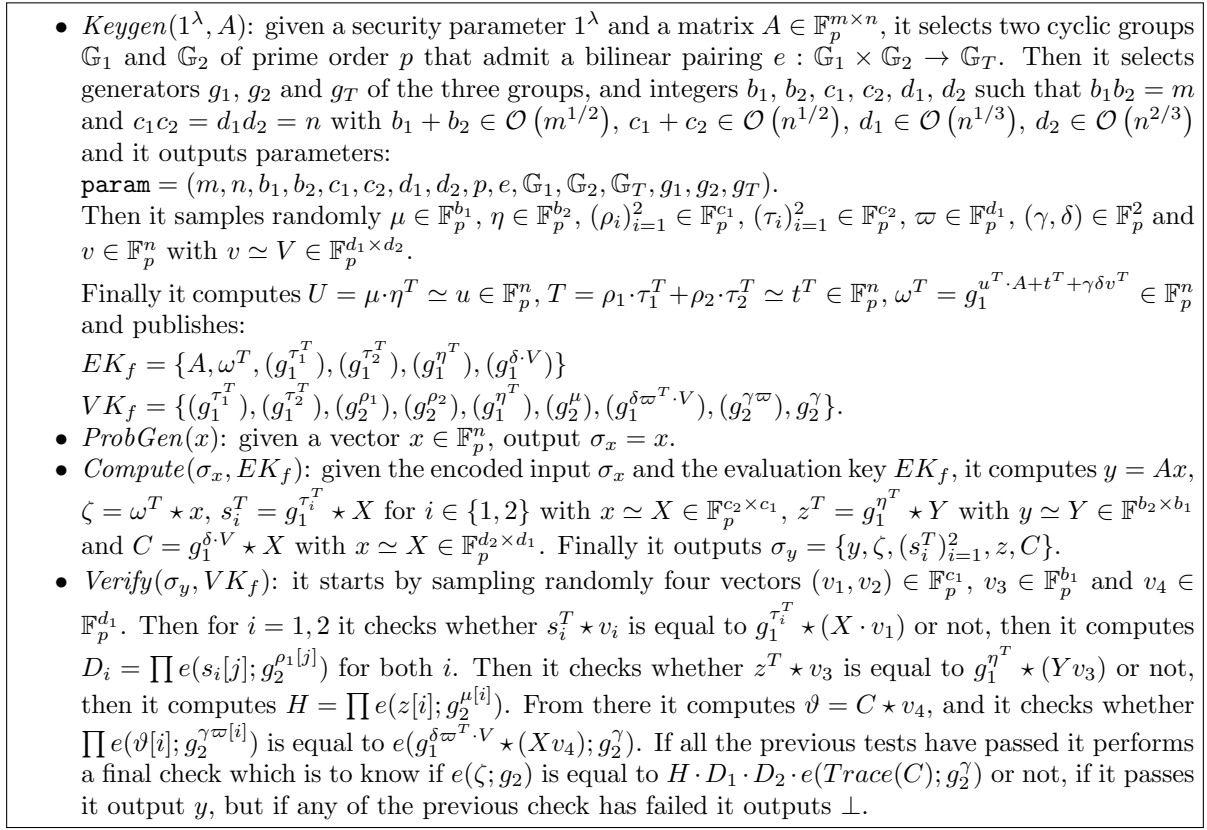
\centering
\fbox{\begin{minipage}{0.95\textwidth}
\begin{compactitem}
\item \emph{Keygen}$(1^\lambda,A)$: given a security parameter $1^\lambda$ and a matrix $A\in\mathbb{F}_p^{m\times n}$, it selects two cyclic groups $\mathbb{G}_1$ and $\mathbb{G}_2$ of prime order $p$ that admit a bilinear pairing $e: \mathbb{G}_1\times  \mathbb{G}_2 \rightarrow  \mathbb{G}_T$. Then it selects generators $g_1$, $g_2$ and $g_T$ of the three groups, and integers $b_1$, $b_2$, $c_1$, $c_2$, $d_1$, $d_2$ such that $b_1b_2=m$ and $c_1c_2=d_1d_2=n$ with $b_1+b_2\in\bigO{m^{1/2}}$, $c_1+c_2\in\bigO{n^{1/2}}$, $d_1\in\bigO{n^{1/3}}$, $d_2\in\bigO{n^{2/3}}$ and it outputs parameters:

$\texttt{param} = (m,n,b_1,b_2,c_1,c_2,d_1,d_2,p,e, \mathbb{G}_1, \mathbb{G}_2, \mathbb{G}_T,g_1,g_2,g_T)$.

Then it samples randomly $\mu\in\mathbb{F}_p^{b_1}$, $\eta\in\mathbb{F}_p^{b_2}$, $(\rho_i)_{i=1}^2\in\mathbb{F}_p^{c_1}$, $(\tau_i)_{i=1}^2\in\mathbb{F}_p^{c_2}$, $\varpi\in\mathbb{F}_p^{d_1}$, $(\gamma, \delta)\in\mathbb{F}_p^2$ and $v\in\mathbb{F}_p^n$ with $v\simeq V\in\mathbb{F}_p^{d_1\times d_2}$.

Finally it computes $U = \mu\cdot\eta^T\simeq u\in\mathbb{F}_p^{n}$, $T=\rho_1\cdot\tau_1^T+\rho_2\cdot\tau_2^T\simeq t^T\in\mathbb{F}_p^n$, $\omega^T=g_1^{u^T\cdot A+t^T+\gamma\delta v^T}\in\mathbb{F}_p^n$ and publishes:

$EK_f=\{A,\omega^T,(g_1^{\tau_1^T}),(g_1^{\tau_2^T}),(g_1^{\eta^T}),(g_1^{\delta\cdot V})\}$

$VK_f=\{(g_1^{\tau_1^T}),(g_1^{\tau_2^T}),(g_2^{\rho_1}),(g_2^{\rho_2}),(g_1^{\eta^T}),(g_2^\mu),(g_1^{\delta\varpi^T\cdot V}),(g_2^{\gamma\varpi}),g_2^\gamma\}$.

\item \emph{ProbGen}$(x)$: given a vector $x\in\mathbb{F}_p^n$, output $\sigma_x = x$.

\item\emph{Compute}$(\sigma_x,EK_f)$: given the encoded input $\sigma_x$ and the evaluation key $EK_f$, it computes $y=Ax$, $\zeta=\omega^T\star x$, $s_i^T=g_1^{\tau_i^T}\star X$ for $i\in\{1,2\}$ with $x\simeq X\in\mathbb{F}_p^{c_2\times c_1}$, $z^T=g_1^{\eta^T}\star Y$ with $y\simeq Y\in\mathbb{F}^{b_2\times b_1}$ and $C=g_1^{\delta\cdot V}\star X$ with $x\simeq X\in\mathbb{F}_p^{d_2\times d_1}$. Finally it outputs $\sigma_y=\{y,\zeta,(s_i^T)_{i=1}^2,z,C\}$.

\item\emph{Verify}$(\sigma_y,VK_f)$: it starts by sampling randomly four vectors $(v_1, v_2)\in\mathbb{F}_p^{c_1}$, $v_3\in\mathbb{F}_p^{b_1}$ and $v_4\in\mathbb{F}_p^{d_1}$. Then for $i=1,2$ it checks whether $s_i^T\star v_i$ is equal to $g_1^{\tau_i^T}\star (X\cdot v_1)$ or not, then it computes $D_i=\prod e(s_i[j];g_2^{\rho_1[j]})$ for both $i$. Then it checks whether $z^T\star v_3$ is equal to $g_1^{\eta^T}\star (Yv_3)$ or not, then it computes $H=\prod e(z[i];g_2^{\mu[i]})$. From there it computes $\vartheta=C\star v_4$, and it checks whether $\prod e(\vartheta[i];g_2^{\gamma \varpi[i]})$ is equal to $e(g_1^{\delta\varpi^T\cdot V}\star(Xv_4);g_2^{\gamma})$. If all the previous tests have passed it performs a final check which is to know if $e(\zeta;g_2)$ is equal to $H\cdot D_1\cdot D_2\cdot e(Trace(C);g_2^\gamma)$ or not, if it passes it output $y$, but if any of the previous check has failed it outputs $\perp$. 
\end{compactitem}
\end{minipage}}
\caption{\label{fig:matrix-vector}Proven publicly delegatable protocol for matrix-vector product}
\sigvspace{-15pt}\end{figure}
Apart from Freivalds's checks and vectorization, we need to use a masking of the
form $u^TA+t^T$ (see Figure~\ref{fig:spmv}) indistinguishable from a random
distribution, but:
\begin{compactenum}
\item We have to add an extra component~$\gamma\delta v^T$ to~$u^TA+t^T$ so that
  it is possible, when proving the reduction to \mbox{co-CDH}, to make up a
  random vector $\omega^T=g^{u^TA+t}$ where the components of $u^TA+t^T$ are
  canceled out. 
  This component cannot be revealed to the Prover nor the Verifier, otherwise its
  special structure could have been taken into account by the reduction. Also
  this component cannot have the rank-$1$ update structure as its has to be a
  multiple of $u^TA+t^T$. Therefore only the protocol of Figure~\ref{fig:gendp}
  can be used to check the dotproduct with~$g^{v^T}$. 
\item To be able to apply the analysis of~\cite[Theorem~3]{Fiore:2012:PVD} while
  allowing fast computations with $t$, we use a special form for $t$, namely: 
  $t^T=\rho^{}_1\tau_{\mathstrut 1}^T+\rho^{}_2\tau_2^T$.
\end{compactenum}
With these modifications we are able to prove the soundness of the protocol in Figure~\ref{fig:matrix-vector}.
\begin{theorem}\label{thm:provenproto}
  Let $A\in\F_p^{m{\times}n}$ whose matrix-vector
  products costs $\mu(A)$ arithmetic operations.
  Protocol~\ref{fig:matrix-vector} is sound under the \mbox{co-CDH} assumption,
  perfectly complete and its number of performed operations is bounded as
  follows: 
\begin{center}\sigsmall
\begin{tabular}{|c||c|c|c|}
\hline
& Preparation  & Prover  & Verifier \\
\hline
$\F_p$  &  $\mu(A){+}\bigO{m{+}n}$  &  $\mu(A)$  &  $\bigO{m{+}n}$ \\
$\G_i$  &  $\bigO{m{+}n}$  &  $\bigO{m{+}n^{4/3}}$  &  $\bigO{\sqrt{m}{+}n^{2/3}}$ \\
Pairings &  0 &  0  &  $\bigO{\sqrt{m}{+}\sqrt{n}}$  \\
\hline
\end{tabular} 
\end{center}
\end{theorem}
The proof of Theorem~\ref{thm:provenproto} is given in
Appendix~\ref{sec:protofig}.
\begin{remark} Fast matrix multiplication can be used for the computation of $C$
	in the protocol of Figure~\ref{fig:matrix-vector}. This decreases the $\bigO{n^{4/3}}$ factor of the Prover to $\bigO{n^{(1+\omega)/3}}$ where $\omega$ is the exponent of
  matrix-matrix multiplication. The currently best known exponent, given
  in~\cite{LeGall:2014:fmm}, is $\omega\leq{}2.3728639$. This immediately
  yields a reduced bound for the Prover of
  $\mu(A)+\bigO{m+n^{1.12428797}}$. 
  This together with Corollary~\ref{cor:gendp}
  can produce a protocol with Prover complexity bounded by
  $\mu(A)+\bigO{m+n^{1+0.12428797a}}$, for any $0<a<1$, while the Verifier
  complexity is $\bigO{n}$ classical operations and $o(n)$ cryptographic
  operations.
\end{remark}

\section{Conclusion and experiments}\label{sec:expes}

We first recall in Table~\ref{tab:complexities} the leading terms of the
complexity bounds for our protocols and those
of~\cite{Fiore:2012:PVD,Blanton:2014:matrix,Elkhiyaoui:2016:ETPV} (that is each
value $x$ in a cell is such that the actual cost is bounded by $x+o(x)$).
\newcommand{\baseField}[1]{\ensuremath{{#1}\cdot{\mathcal F}}}
\newcommand{\expoNent}[1]{\ensuremath{{#1}\cdot{\mathcal G}}}
There, we denote the base field operations by \baseField{},
the cryptographic group exponentiations or pairing operations by \expoNent{},
and the cost of a product of the matrix $A\in\F_p^{m{\times}n}$ by a vector is 
$\mu(A)$.
\begin{table}[htb]\centering
\caption{Leading terms for the time and memory complexity bounds (exchange of
  $A$, $x$ and $y$ excluded).}\label{tab:complexities}
\begin{tabular}{|c||c|c|c|} 
\hline
Scheme  & \cite{Fiore:2012:PVD} & \cite{Blanton:2014:matrix} & \cite{Elkhiyaoui:2016:ETPV} \\
\hline
Mode  &  Public verif.  &  Public  deleg. &  Public  deleg. \\
\hline
Preparator (KeyGen)	&  $\baseField{2mn}+\expoNent{mn}$ 
			&  $-$  
			&  $\baseField{2mn}+\expoNent{2mn}$ \\
Trustee (ProbGen)	&  $\baseField{2(m+n)}+\expoNent{2m}$  
			&  $\baseField{mn}+\expoNent{(2m+n)}$  
			&  $\expoNent{n}$  \\
Prover (Compute)	&  $\baseField{\mu(A)}+\expoNent{2mn}$  
			&  $\baseField{\mu(A)}+\expoNent{2mn}$ 
			&  $\baseField{\mu(A)}+\expoNent{2mn}$  \\
Verifier		&  $\expoNent{2m}$  
			&  $\expoNent{2m}$ 
			&  $\expoNent{m}$  \\
\hline
Extra storage		& $\bigO{mn}$ & $\bigO{mn}$ & $\bigO{mn}$ \\
Extra communications	& $\bigO{m}$ & $\bigO{m}$ & $\bigO{1}$\\ 
\hline
\end{tabular}

\smallskip

\begin{tabular}{|c||c|c|}
\hline
Scheme  & Figure~\ref{fig:spmv} & Figure~\ref{fig:matrix-vector} \\
\hline
Mode  &  Public verif. &  Public deleg.  \\
\hline
Preparator (KeyGen)  	&  $\baseField{(\mu(A)+n)}+\expoNent{n}$  
			&  $\baseField{(\mu(A)+m+5n)}+\expoNent{2n}$  \\
Trustee (ProbGen)  	&  $\baseField{2(m+n+1)}+\expoNent{1}$  
			&  $0$  \\
Prover (Compute)  	&  $\baseField{\mu(A)}+\expoNent{2n}$  
			&  $\baseField{\mu(A)}+\expoNent{(2n^{4/3}+m)}$  \\
Verifier 		&  $\expoNent{1}$ 
			&  $\baseField{(2m+4n)}+\expoNent{(6\sqrt{m}+2n^{2/3})}$\\
\hline
Extra storage		& $\bigO{m+n}$ & $\bigO{n}$ \\
Extra communications	& $\bigO{1}$ & $\bigO{n^{2/3}+\sqrt{m}}$\\ 
\hline
\end{tabular}\end{table}
We see that our protocols are suitable to sparse or structured matrix-vector
multiplication as 
they never require $\bigO{mn}$ operations but rather $\mu(A)$. Moreover, we see
that most of the Verifier's work is now in base field operations were it was
cryptographic operations for previously known protocols.
\begin{table}[htb]\centering
\caption{Matrix-vector multiplication public verification over a 256-bit finite
  field with different protocols on a i7 @3.4GHz.}\label{tab:expes} 
\sigsmall
\begin{tabular}{|l||r|r|r|r|r||r|r|r|r|}
\hline
& \multicolumn{5}{|c||}{$1000{\times}1000$}  &
\multicolumn{4}{|c|}{$2000{\times}2000$} \\
\cline{2-10}
& \multicolumn{1}{|c|}{\cite{Setty:2012:pepper}} &
\multicolumn{1}{|c|}{\cite{Fiore:2012:PVD}} &
\multicolumn{1}{|c|}{\cite{Blanton:2014:matrix}} &
\multicolumn{1}{|c|}{\cite{Elkhiyaoui:2016:ETPV}} & 
\multicolumn{1}{|c||}{Fig.~\ref{fig:matrix-vector}} & 
\multicolumn{1}{|c|}{\cite{Fiore:2012:PVD}} &
\multicolumn{1}{|c|}{\cite{Blanton:2014:matrix}} &
\multicolumn{1}{|c|}{\cite{Elkhiyaoui:2016:ETPV}} & 
\multicolumn{1}{|c|}{Fig.~\ref{fig:matrix-vector}} \\
\hline
KeyGen	&141.68s& 152.62s	& -		& 154.27s 	& {\bf 0.80s}
	& 615.81s	& -		& 612.72s 	& {\bf 1.75s}\\
ProbGen	&-	&  1.25s		& 2.28s		& 2.30s		& -     
	& 2.13s	& 4.98s		& 4.56s		& -     \\
$Ax=y$	&20.14s	& \bf 0.19s	& \bf 0.19s	& \bf 0.19s	& \bf 0.19s
	& \bf 0.78s & \bf 0.78s	& \bf 0.78s	& \bf 0.78s  \\
Compute	&188.60s& 273.06s 	& 433.88s	& 271.03s	& {\bf 2.26s}
	& 1097.96s & 1715.46s	& 1079.71s	& {\bf 5.37s}\\
Verify	&2.06s	& 26.62s		& 27.56s		& {\bf 0.33s}	& 0.90s
	& 52.60s	& 55.79s		& {\bf 0.62s}	& 1.19s\\
\hline
\multicolumn{6}{c}{}\\
\end{tabular}
\sigvspace{-5pt}
\begin{tabular}{|l||r|r|r|r||r|r|r|r|}
\hline
& \multicolumn{4}{|c||}{$4000{\times}4000$}  &
\multicolumn{4}{|c|}{$8000{\times}8000$} \\
\cline{2-9}
& 
\multicolumn{1}{|c|}{\cite{Fiore:2012:PVD}} &
\multicolumn{1}{|c|}{\cite{Blanton:2014:matrix}} &
\multicolumn{1}{|c|}{\cite{Elkhiyaoui:2016:ETPV}} & 
\multicolumn{1}{|c||}{Fig.~\ref{fig:matrix-vector}} & 
\multicolumn{1}{|c|}{\cite{Fiore:2012:PVD}} &
\multicolumn{1}{|c|}{\cite{Blanton:2014:matrix}} &
\multicolumn{1}{|c|}{\cite{Elkhiyaoui:2016:ETPV}} & 
\multicolumn{1}{|c|}{Fig.~\ref{fig:matrix-vector}} \\
\hline
KeyGen	& 2433.10s	& -		& 2452.98s	& {\bf 4.89s}
	& 9800.42s	& -		& 9839.26s 	& {\bf 15.64s}\\
ProbGen	& 3.81s		& 13.29s		& 9.24s		& -
	& 7.41s		& 43.44s		& 18.46s		& -     \\
$Ax=y$	& {\bf 3.28s}	& {\bf 3.28s}	& {\bf 3.28s}	& {\bf 3.28s}
	& \bf  13.30s	& \bf 13.30s	& \bf 13.30s	& {\bf 13.30s} \\
Compute	& 4360.43s	& 6815.40s	& 4329.46s	& {\bf 13.76s}
	& 17688.69s 	& 27850.90s	& 17416.38s	& {\bf 37.00s}\\
Verify	& 103.14s	& 107.99s	& {\bf 1.20s}	& 1.65s
	& 211.07s	& 220.69s	& 2.37s 		& {\bf 2.25s}\\
\hline
\end{tabular}
\sigvspace{-5pt}
\end{table}
As shown in Table~\ref{tab:expes} and in Figure~\ref{fig:provprotexpes}, this is
very useful in practice, even for dense matrices. For these experiments we 
compare with our own implementations of the protocols 
of~\cite{Fiore:2012:PVD,Blanton:2014:matrix,Elkhiyaoui:2016:ETPV} 
over the PBC library\footnote{\url{https://crypto.stanford.edu/pbc},
  version~0.5.14}~\cite{Lynn:2010:PBC} for the pairings and the FFLAS-FFPACK
library\footnote{\url{http://linbox-team.github.io/fflas-ffpack},
  version~2.2.2}~\cite{jgd:2008:toms} for the exact linear algebra over finite
fields (C++ source files are available on request via the PC and will be
publicly posted on our web site if the paper is accepted).  
\begin{figure}[!htb]
\includegraphics[width=\textwidth]{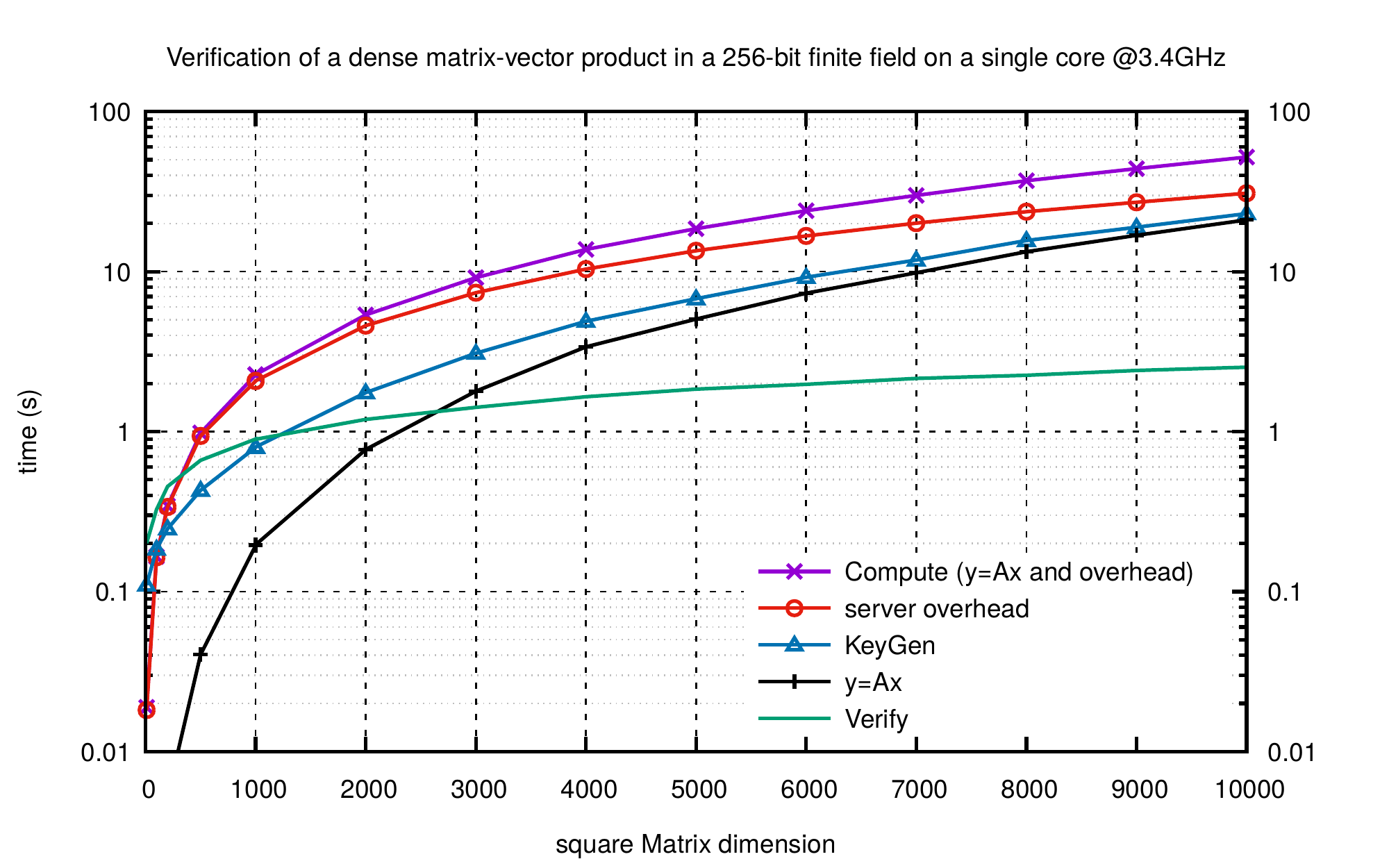}
\caption{Protocol~\ref{fig:matrix-vector} performance.}\label{fig:provprotexpes}
\sigvspace{-15pt}\end{figure}
We used randomly generated dense matrices and vectors and to optimize the costs
(pairings are more expensive than exponentiations), we also chose the following
parameters for the vectorizations:
$b_1=\lceil\sqrt{m}/10\rceil{}$; $b_2=\lceil{}10\sqrt{m}\rceil{}$;
$c_1=\lceil{}\sqrt{n}/10\rceil{}$; $c_2=\lceil{}10\sqrt{n}\rceil{}$;
$d_1=\lceil{}n^{1/3}/3\rceil{}$; $d_2=\lceil{}3n^{2/3}\rceil{}$.
We indeed chose a type 3 pairing over a
Barreto-Naehrig curve \cite{Barreto2006} based on a 256-bits prime field, which should guarantee 128 bits of security. 
First, with $\F_p$ the $256$-bits prime
field\footnote{{\tiny $p=57896044618658115533954196422662521694340972374557265300857239534749215487669$}},
$\G_1$ is the group of $\F_p$-rational points $E(\F_p)$ with parameters: 
$\G_1~(E): y^2 = x^3 + 6$, modulo $p$.
Second, $\G_2$ is a subgroup of a sextic twist of $E$ defined over
$\F_{p^2}$ denoted $E'(\F_{p^2})$ with parameters\footnote{\tiny
  $a_0{=}52725052272451289818299123952167568817548215037303638731097808561703910178375$, $a_1{=}39030262586549355304602811636399374839758981514400742761920075403736570919488$}:
$\G_2~(E'): y^2 = x^3 + 6e$, $\F_{p^2}\cong \F_p[X]/(X^2-2)$, $e = a_0 + a_1
X\in\F_{p^2}$.
The third group $\G_T$ is then a subgroup of the multiplicative group of
the field $\F_{p^{12}}$.
This curve is reasonably well-suited to our needs and is supported by the PBC
library. 
We used it for our own implementation of the three
protocols~\cite{Fiore:2012:PVD,Blanton:2014:matrix,Elkhiyaoui:2016:ETPV} as
well as for ours, Section~\ref{sec:fully} and Figure~\ref{fig:provprot}. 

In the first set of timings of Table~\ref{tab:expes}, we also compare the latter
protocols with a compiled verifiable version obtained via the Pepper
software\footnote{\url{https://github.com/pepper-project/pepper}, git:~fe3bf04}~\cite{Setty:2012:pepper}.
This software uses a completely different strategy, namely that of compiling a C
program into a verifiable one. We added the timings for $n=1000$ as a comparison,
but the Pepper compilation thrashed on our 64 GB machine for $n\geq{}2000$.

In terms of Prover time, we see that our protocols are between two to three
orders of magnitude faster than existing ones (further evidence is given in
Table~\ref{tab:speedup}, Appendix~\ref{sec:protofig}). 
Moreover, overall we see that
with the new protocol, the data preparation (KeyGen) is now very close to a
single non-verified computation and that the work of the Prover can be less
than three times that of a non-verified computation (note first, that in both
Table~\ref{tab:expes} and Figure~\ref{fig:provprotexpes}, the ``Compute'' fields
include the computation of $y=Ax$, and, second,
that the Prover overhead being asymptotically faster than the compute time, this
latter overhead is rapidly amortized).
Finally, only the protocol of~\cite{Elkhiyaoui:2016:ETPV} did exhibit a
verification step faster than the computation itself for size $2000{\times}2000$
whereas, as shown in Figure \ref{fig:provprotexpes}, our protocol achieves
this only from size $3000{\times}3000$.
However, we see that we are competitive for larger matrices. Moreover, as shown
by the asymptotics of Theorem~\ref{thm:provenproto}, our overall performance
outperforms all previously known protocols also in practice, while keeping an
order of magnitude faster Verification time. 



\appendix

 \section{Fiore and Gennaro's protocol}\label{sec:fg}
 For the sake of completeness, we present here the original protocol for
 matrix-vector verification in~\cite{Fiore:2012:PVD}, but with our rank-one
 update view. It stems from the fact that if $s,t,\rho,\tau$ are
 randomly generated vectors 
 then the function $g^{M[i,j]}$, where $M=s\cdot{}t^T+\rho\cdot\tau^T$, is a
 pseudorandom function~\cite[Theorem~3]{Fiore:2012:PVD}, provided that the
 \emph{Decision Linear} assumption~\cite[Definition~3]{Fiore:2012:PVD} holds (a
 generalization of the External Diffie-Hellman assumption for pairings).
 \begin{itemize}
 \item \emph{KeyGen}: 
   for $A\in\F_p^{m{\times}n}$, generate 3 multiplicative groups
   $(\G_1,\G_2,\G_T)$ of prime order $p$, with $g_1$ generating $\G_1$
   (resp. $g_2$ generating $\G_2$), and a bilinear map
   $e:\G_1{\times}\G_2\rightarrow\G_T$.
   Generate $2(m+n)$ secret random values $s\in\F_p^m$, $t\in\F_p^n$,
   $\rho\in\F_p^m$, $\tau\in\F_p^n$.
   Compute $W[i,j]=g_1^{\alpha A[i,j]+s[i]t[j]+ \rho[i]\tau[j]}\in\G_1$, give $W$ to the server and publish
   $a=e(g_1^{\alpha},g_2)\in\G_T$.
 \end{itemize}
 Let $x\in\F_p^n$ be a query vector
 \begin{itemize}
 \item \emph{ProbGen}: 
   compute $\text{VK}_x\in\G_T^m$, such that $d=t^T\cdot{}x\in\F_p$,
   $\delta=\tau^T\cdot{}x\in\F_p$, and $\text{VK}_{x}[i]=e(g_1^{s[i]d+\rho[i]\delta};g_2)\in\G_T$.
 \item \emph{Compute}: 
   compute $y=Ax$ and $z=W\star{}x\in\G_1^m$ (that is
   $z[i]=\prod_{j=1}^n W[i,j]^{x[j]}$).
 \item \emph{Verify}:
   check that $e(z[i];g_2)=a^{y[i]} \text{VK}_{x}[i]$, for all $i=1,\ldots,m$.
 \end{itemize}
In this protocol, the flow of communications is as follows:
\begin{enumerate}
\item Preparator: secret random $\alpha\in\F_p$, $s\in\F_p^m$, $t\in\F_p^n$, $\rho\in\F_p^m$, $\tau\in\F_p^n$.
then $W=g_1^{\alpha A+s\cdot{}t^T+\rho\cdot{}\tau^T}$
\item Preparator to Prover: $A\in\F_p^{m{\times}n}$, $W\in\G_1^{m{\times}n}$.
\item Preparator to Trustee: $s$, $\rho$, $t$, $\tau$, in a secure channel.
\item Preparator publishes and signs $a=e(g_1;g_2)^{\alpha}\in\G_T$.
\item Verifier to both Prover and Trustee: $x\in\F_p^n$.
\item Trustee publishes and signs $\text{VK}_{x}\in\G_T^m$ such that :
  $\text{VK}_{x}=e(g_1;g_2)^{s\cdot{}(t^T\cdot{}x)+\rho\cdot{}(\tau^T\cdot{}x)}$.
\item Prover to Verifier: $y\in\F_p^m$, $z\in\G_1^m$. 
\item Verifier public verification: $e(z;g_2)\checks{}a^{y}\text{VK}_{x}$
  (com\-po\-nent-wise in $\G_T^m$).
\end{enumerate}

 This protocol is sound, complete and publicly verifiable. It however uses many
 costly exponentiations and pairings operations that renders it inefficient in
 practice: even though the Client and Trustee number of operations is linear in
 the vector size, it takes still way longer time that just computing the
 matrix-vector product in itself, as shown in the experiment
 Section~\ref{sec:expes}. 

 In this paper, our aim was first to adapt this protocol to the
 sparse/structured case, and, second, to reduce the number of
 cryptographic operations in order to obtain a protocol efficient in
 practice.

\section{Probabilistic verification and the random oracle model}\label{sec:probaverif}
First we recall that private verification is very fast and does not require any
cryptographic routines. Then we show that this allows to obtain a very efficient
protocol in the random oracle model, but for a fixed number of inputs.

\subsection{Private verification}\label{ssec:freivalds}
Without any recourse to cryptography, it is well known how to privately verify a
matrix-vector multiplication. The idea is to use Freivalds
test~\cite{Freivalds:1979:certif}, \emph{on the left}, provided that
multiplication by the transpose matrix is possible:\\

\noindent\fbox{\begin{minipage}{\linewidth - 3\fboxsep}\sigsmall
\begin{itemize}
\item Verifier to Prover: $A$, $x_i\in\F_p^n$, for $i=1,\ldots,k$.
\item Prover to Verifier: $y_i\in\F_p^m$, for $i=1,\ldots,k$.
\item Verifier verification: random $u\in\F_p^m$, then $w^T = u^T\cdot{}A$,
and finally check, for $i=1,\ldots,k$, that
$w^T\cdot{}x_i\checks{}u^T\cdot{}y_i$ in $\F_p$.
\end{itemize}
\end{minipage}}\\

On the one hand, this protocol uses only classical arithmetic and is adaptable
to sparse matrices, that is when a matrix vector product costs 
$\mu(A)$ operations with $\mu(A)<2mn$ (this is the case for instance if the
matrix is not structured but is sparse with $\mu(A)/2<mn$ non-zero elements).  
Indeed, in the latter case, the cost for the Prover is $k\mu(A)$, where the
cost for the Verifier is $\mu(A)+4kn$.

On the other hand, the protocol has now Freivalds
probability of revealing an error in any of the $y_i$: $1-1/p$, if $\F_p$ is of
cardinality~$p$ (or $1-1/p^\ell$ if $u$ is chosen in an extension of degree
$\ell$ of $\F_p$).

\subsection{Public verification in the random oracle model}\label{ssec:oracle}
Using Fiat-Shamir heuristic~\cite{Fiat:1986:Shamir}, the privately
verifiable certificate of Section~\ref{ssec:freivalds} can be simulated
non-interactively: uniformly sampled 
random values produced by the Verifier are replaced by cryptographic hashes (to prove
security in the random oracle model) of the input and of previous messages in
the protocol.   
Complexities are preserved, as producing cryptographically strong
pseudo-random bits by a cryptographic hash function (e.g., like the
extendable output functions of the SHA-3 family defined
in~\cite{Bertoni:2010:sponge,SHAKE}), is linear in the size of both its input
and output (with atomic operations often even faster than finite field ones):\\ 

\noindent\fbox{\begin{minipage}{\linewidth - 3\fboxsep}\sigsmall
\begin{itemize}
\item Preparator to Prover: $A\in\F_p^{m{\times}n}$.
\item Verifier to Prover: $x_i\in\F_p^n$, for $i=1,\ldots,k$.
\item Prover to Verifier: $y_i\in\F_p^m$, for $i=1,\ldots,k$.
\item Verifier to Trustee: all the $x_i$ and $y_i$.
\item Trustee publishes and signs both $u\in\F_p^m$ and $w\in\F_p^n$ such that:
  $u=Hash(A,x_1,\ldots,x_k,y_1,\ldots,y_k)\in\F_p^m$, then $w^T = u^T\cdot{}A\in\F_p^n$. 
\item Verifier public verification: $w^T\cdot{}x_i\checks{}u^T\cdot{}y_i$ in $\F_p$.
\end{itemize}
\end{minipage}}\\

\bigskip

There is absolutely no overhead for the Prover;
the cost for the Trustee is a single matrix-vector product for any $k$ plus a
cost linear in the input size;
and the cost for the Verifier is $\bigO {nk}$. Using Fiat-Shamir heuristic this
gives a possibility for an afterwards public verification (that is after the
computations), but this not possible to test new vectors once $u$ has been
revealed.

\section{Proof of Theorem~\texorpdfstring{\ref{thm:sparsecdh}}{1}}\label{ssec:proofsparse}
We here give the proof of Theorem~\ref{thm:sparsecdh},
page~\pageref{thm:sparsecdh}, recalled hereafter.
\newcounter{oldtheorem}
\setcounterref{theorem}{thm:sparsecdh}
\addtocounter{theorem}{-1}
\begin{theorem} The protocol of Figure~\ref{fig:spmv} is
  perfectly complete and sound under the co-Computational Diffie-Hellman Problem
  assumption.
\end{theorem}\setcounter{theorem}{\theoldtheorem}
\begin{proof}
For the correctness, we have that:
$\zeta_i =g_1^{(u^TA+t^T)\cdot{}x_i} =g_1^{u^T\cdot{}y_i+t^T\cdot{}x_i}=g_1^{h_i}g_1^{d_i}=g_1^{h_i+d_i}$.
Then, by bilinearity,
$e(\zeta_i;g_2)=e(g_1;g_2)^{h_i+d_i}=\eta_i$.

For the soundness,
a malicious Prover can guess the correct output values, but this happens once in
the number of elements of $\G_T$. 
Otherwise he could try to guess some matching $h_i$ and
$d_i$, but that happens less than one in the number of elements
of $\F_p$. 
Finally, the Prover could produce directly $\zeta_i$. 
Suppose then it is possible to pass our verification scheme for some $A$, $x$
and $y'\neq{}y=Ax$. Then without loss of generality, we can suppose that the
first coefficients of both vectors are different, $y'[1]\neq{}y[1]$ (via row
permutations) and that $y'[1]-y[1]=1$ (via a scaling). 

Take a co-computational Diffie-Hellman problem $(g_1^c,g_2^d)$, where $g_1^{cd}$
is unknown.  Then denote by $a=e(g_1^c;g_2^d)=e(g_1^{cd};g_2)$ and 
consider the vector $z^T=[a,e(1;1),\ldots,e(1;1)]$. Compute
$\chi^T=z^T\star{}A$. The latter correspond to $\chi^T=e(g_1^{u^TA};g_2)$ for (a
not computed) $u^T=[cd,0,\ldots,0]$. 
Now randomly choose $\psi^T=[\psi_1,\ldots,\psi_n]$ and compute 
$\omega^T=g_1^{\psi^T}$.
Compute also the 
vector $\phi^T=e(\omega^T;g_2)/\chi^T$ coefficient-wise. The latter correspond 
to $\phi^T=e(g_1^{t^T};g_2)$ for $t^T=\psi^T-u^TA$.
Finally, compute
$\zeta=g_1^{\psi^T\cdot{}x}$
(indeed, then 
$\mu=e(\zeta;g_2)=e(g_1^{\psi^T\cdot{}x};g_2)=
\eta=e(g_1^{u^T\cdot{}y};g_2)e(g_1^{t^T\cdot{}x};g_2)$, 
that is 
$\eta=(\chi^T\star{}x)(\phi^T\star{}x)$ is actually
$\eta=(z^T\star{}y)(\phi^T\star{}x)$).
Now, if it is possible to break the scheme, then it is possible to compute
$\zeta'$ that will pass the verification for $y'$ as $Ax$,
that is $e(\zeta';g_2)=(z^T\star{}y')(\phi^T\star{}x)$.  
Let $h=u^Ty$, $d=t^Tx$ and $h'=u^Ty'$.
Then $e(\zeta;g_2)=e(g_1^h;g_2)e(g_1^d;g_2)$ and 
$e(\zeta';g_2)=e(g_1^{h'};g_2)e(g_1^d;g_2)$. But
$h'-h=u^T(y'-y)=cd(y'[1]-y[1])=cd$ by construction. 
Therefore $\zeta'/\zeta=g_1^{cd}$, as $e$ is non-degenerate, and the \mbox{co-CDH} is
solved. 
\end{proof}

\section{Small fields}\label{ssec:smallfields}
The protocol of Figure~\ref{fig:spmv} is quite efficient. 
We have made experiments with randomly generated dense matrices and vectors with
the PBC library\footnotemark[1] for the pairings and the FFLAS-FFPACK
library\footnotemark[2] for the exact linear algebra over finite
fields.  
For instance, it is shown in Table~\ref{tab:spmv}, that for a $8000{\times}8000$
matrix over a field of size $256$ bits, 
the protocol is highly practical:  first, if the base field and the group orders
are of similar sizes, the verification phase is very efficient; 
second, the overhead of computing $\zeta$ for the server is quite negligible and
third, the key generation is dominated by the computation of one matrix-vector
product.

\begin{table}[htb]\centering
\caption{Verification of a $8000{\times}8000$ matrix-vector
  multiplication with different field sizes via the protocol in
  Figure~\ref{fig:spmv} on a single core @3.4GHz.}\label{tab:spmv}
\begin{tabular}{|r||r|r||r|r|r||r|r|r||r|}
\hline
\multirow{2}{*}{Field size}& \multirow{2}{*}{$|\G|$}& \multirow{2}{*}{Security}
& \multicolumn{3}{|c|}{KeyGen}
& \multicolumn{3}{|c|}{Compute}
& \multirow{2}{*}{Verify} \\
\cline{4-9}
& & & Total & $u^TA$ & overhead & Total & $y=Ax$ & overhead & \\
\hline
256& 256 & 128 & 13.65s &12.34s&1.22s &15.72s &13.46s&2.26s&0.03s\\
10& 322 & 128 & 1.96s&0.05s&1.81s&0.22s&0.09s&0.13s&0.04s\\
\hline
\end{tabular}
\end{table}

Differently, if the base field is small, say machine word-size, then having to
use cryptographic sizes for the group orders can be penalizing for the Key
Generation: multiplying a small field matrix $A$ with a large field
vector $u^T$ is much slower than $y=Ax$ with $x$ and $A$ small.
First of all, the computations must be compatible. For this, one possibility is
to ask and verify instead for $y=Ax$ over $\Z$ and then to let the Verifier
compute $y\mod{}p$ for himself.
There, to reduce the overhead of computing $u^TA$, one can instead select the
$m$ values of the vector $u$ as $u_\ell=\alpha r_i s_j$ with
$\ell=i\lceil\sqrt{m}\rceil{}+j$ for $\alpha$ a randomly chosen large value and
$r_i,s_j$ some randomly chosen small values. 
Indeed then $u^T A$ can be computed by first performing $(r s^T)A$
via $\bigO{\sqrt{m}}$ matrix-vector computations with $s$ (or a
$\sqrt{m}{\times}n\sqrt{m}$ matrix-vector multiplication) followed by
$\bigO{n\sqrt{m}}$ multiplications by $r$ (or a $n{\times}\sqrt{m}$
matrix-vector multiplication) where $s_j$ and $r_i$ are small values. 
Then it remains only to multiply a vector of small values by $\alpha$. 
We have traded $\bigO{mn}$ operations with large values for 
$\bigO{\sqrt{m}n\sqrt{m}+n\sqrt{m}}$ operations with small values and $\bigO{n}$
with large values. 

Now, in order for the values to remain correct over $\Z$, the value of 
$(u^TA+t^T)x$ must not overflow. For this, one must choose a group order
larger than $mnp^4$ (for $(r s^T)Ax$).
Now the security is not anymore half the size of the group order but potentially
half the size of the set from which $t^T$ is selected, that is at most the group
order size minus that of $np$ (for $t^Tx$). 
To be conservative we even propose,
as an estimated security of the obtained protocol, to consider only half the
size of $\alpha$ (that is the size of the group order minus that of $mnp^4$). 
In terms of efficiency, the
improvement is shown in Table~\ref{fig:spmv}, last row. 
On the one hand, the key generation is now dominant and can be amortized only
after about $10$ matrix-vector multiplications. 
On the other hand, the verification time starts to be faster
than the computation time. 
This is also shown in Figure~\ref{fig:smallfields} where the equivalent of
the last row in Table~\ref{tab:spmv} is shown for different matrix dimensions.
\begin{figure}[htb]
\includegraphics[width=\linewidth]{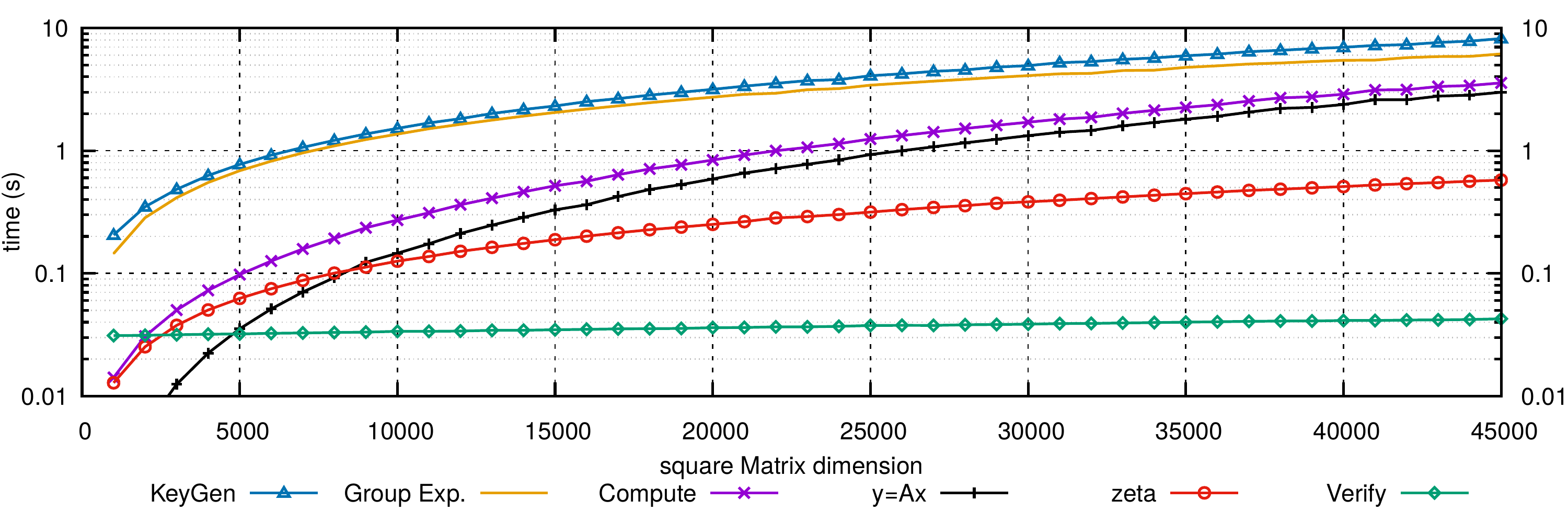}
\caption{Trustee-helped Verification of a dense matrix-vector product in a 10-bits finite field on a single core @3.4GHz.}\label{fig:smallfields}
\end{figure}


\section{Proven publicly delegatable protocol with negligible
  cryptographic operations}\label{sec:protofig}

We first give the proof of Theorem~\ref{thm:provenproto},
page~\pageref{thm:provenproto}, recalled hereafter,
for the correctness, soundness and complexity of the protocol in
Figures~\ref{fig:matrix-vector} and~\ref{fig:provprot}.
The flow of exchanges within our protocol is also illustrated in
Figure~\ref{fig:provprot}. 

\begin{figure}[htbp]\center
\hspace*{-1.8cm}\input{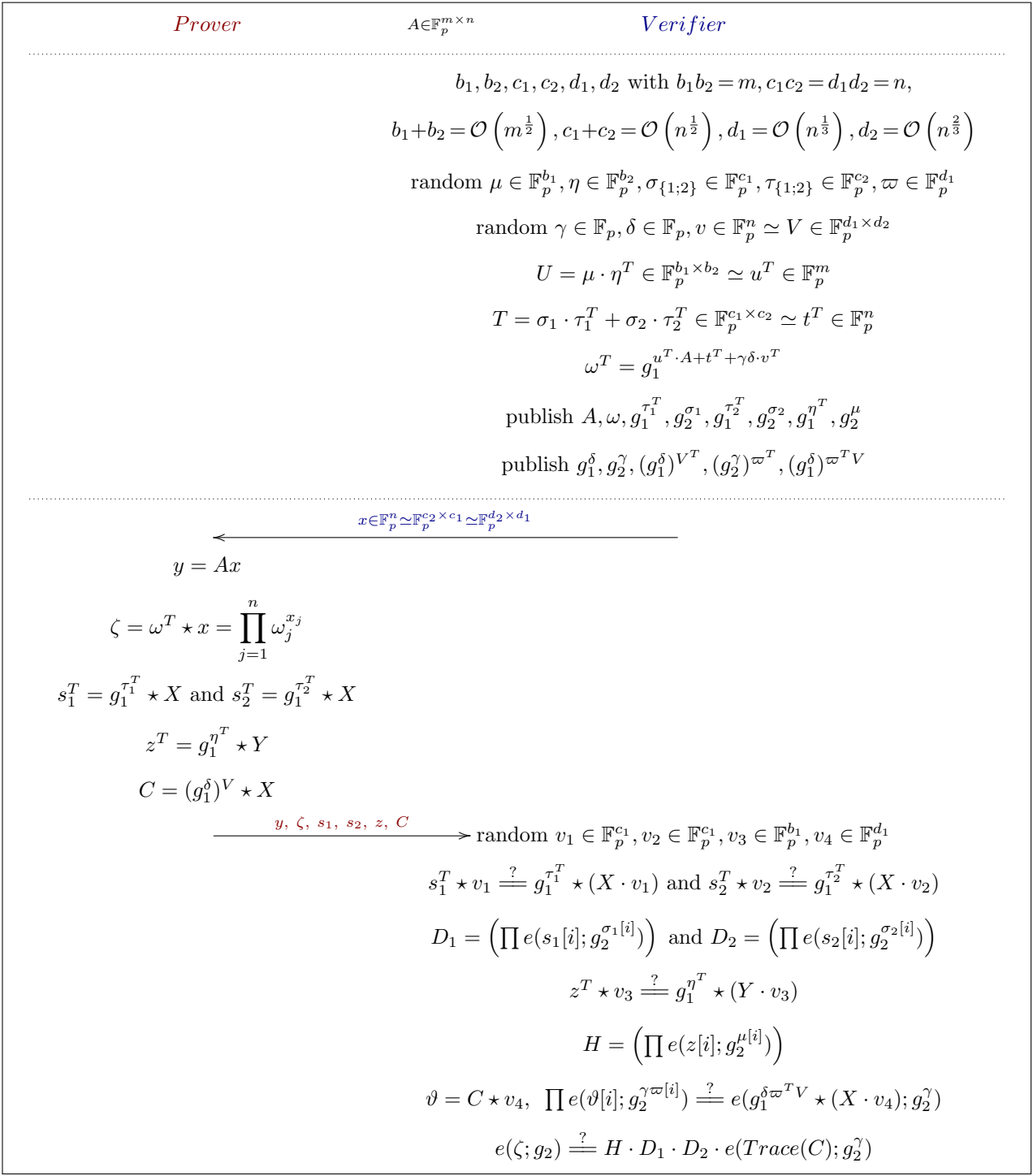}
	\caption{Exchanges in the proven publicly delegatable protocol with negligible cryptographic operations of Figure~\ref{fig:matrix-vector}.}\label{fig:provprot}
\end{figure}

\setcounter{oldtheorem}{\thetheorem}
\setcounterref{theorem}{thm:provenproto}
\addtocounter{theorem}{-1}
\begin{theorem}
  Let $A\in\F_p^{m{\times}n}$ whose matrix-vector
  products costs $\mu(A)$ arithmetic operations.
  Protocol~\ref{fig:matrix-vector} is sound under the \mbox{co-CDH} assumption. It is
  also perfectly complete and its number of performed operations is bounded as
  follows: 
\begin{center}\sigsmall
\begin{tabular}{|c||c|c|c|}
\hline
& Preparation  & Prover  & Verifier \\
\hline
$\F_p$  &  $\mu(A){+}\bigO{m{+}n}$  &  $\mu(A)$  &  $\bigO{m{+}n}$ \\
$\G_i$  &  $\bigO{m{+}n}$  &  $\bigO{m{+}n^{4/3}}$  &  $\bigO{\sqrt{m}{+}n^{2/3}}$ \\
Pairings &  0 &  0  &  $\bigO{\sqrt{m}{+}\sqrt{n}}$  \\
\hline
\end{tabular} 
\end{center}
\end{theorem}\setcounter{theorem}{\theoldtheorem}
\begin{proof}
  Completeness stems again directly from Equation~(\ref{eq:tracerk1}).
  
  For the complexity bounds, we fix $b_1b_2\geq{}m$ and $b_1+b_2=
  \bigTheta{\sqrt{m}}$ (usually, pairing operations are costlier than
  group operations, therefore a good practice could be to take $b_1<b_2$ and we,
  for instance, often have used $b_2=100b_1$ with $b_1b_2\approx{} m$ which gave
  us a speed-up by a factor of~$5$),
  $c_1c_2\geq{}n$ and $c_1+c_2= \bigTheta{\sqrt{n}}$, and finally
  $d_1=\bigO{m^{1/3}}$ and $d_2=\bigO{m^{2/3}}$.
  For the Prover, we then have that $y$ obtained in $\mu(A)$ operations; 
  $\zeta$ in $\bigO{n}$; $s_1$, $s_2$ and $z$ computations are bounded by
  $\bigO{n+m}$ where $C$ thus requires $\bigO{n^{4/3}}$ operations.
  The cost for the preparation is $\bigO{m+n}$ for $U$, $T$ and $\varpi^TV$.
  $\omega$ requires $\mu(A)+2m$ classical operations and $\bigO{m}$ group
  operations. $(g_1^{\delta})^{V^T}$ requires $\bigO{n}$ group operations while
  $g_1^{\tau_i}$, $g_1^{\eta_i}$, $g_2^{\rho_i}$, and $g_2^{\mu_i}$,  require
  $\bigTheta{\sqrt{m}+\sqrt{n}}$ operations, more than for
  $(g_2^{\gamma})^{\varpi^T}$ and $(g_1^{\delta})^{\varpi^TV}$. 
  The complexity for the Verifier is then dominated by $\bigO{n^{2/3}}$ operations
  to check~$C$, $\bigO{n}$ classical operations for $Y\cdot{}v_3$ and
  $\bigO{\sqrt{n}}$ pairing operations.
  
  Finally for the soundness, assume that there is an adversary $\mathcal{A}$ that breaks the soundness of our protocol with non-negligible advantage $\epsilon$ for a matrix $A\in\mathbb{F}_p^{m\times n}$. In the following we will prove how an adversary $\mathcal{B}$ can use adversary $\mathcal{A}$ to break the \mbox{co-CDH} assumption with non-negligible advantage $\epsilon'\simeq\epsilon$.
  Let assume that $\mathcal{B}$ was given a \mbox{co-CDH} sample $(L=g_1^a,R=g_2^b)$.  
  First $\mathcal{B}$ simulates the soundness experiment to adversary $\mathcal{A}$ in the following manner:
  when $\mathcal{A}$ calls the oracle $\mathcal{O}_{KeyGen}$, adversary $\mathcal{B}$ first chooses integers, $b_1$, $b_2$, $c_1$, $c_2$, $d_1$, and $d_2$ such that $m=b_1b_2$ and $n=b_1b_2=d_1d_2$. Then it generates random vectors
  $\mu_0\in\F_p^{b_1},\eta_0\in\F_p^{b_2},\rho_{01}\in\F_p^{c_1},\tau_{01}\in\F_p^{c_2},\rho_{02}\in\F_p^{c_1},\tau_{02}\in\F_p^{c_2}$, $\varpi\in\F_p^{d_1}$ 
  and a value $r\in\F_p$.
  We let $u_0$ be the vector representation of $\mu_0\cdot{}\eta_0^T$ and $t_0$
  that of $\rho_{01}\cdot{}\tau_{01}^T+\rho_{02}\cdot{}\tau_{02}^T$. 
  We also let $v=-(A^Tu_0+t_0)\in\F_p^n$.
  Finally, $\mathcal{B}$ forms 
  $\omega^T=L^{r\cdot{}v^T}$;
  $g_1^\eta=L^{\eta_0}$, $g_2^\mu=R^{\mu_0}$;
  $g_1^{\tau_1}=L^{\tau_{01}}$, $g_2^{\rho_1}=R^{\rho_{01}}$;
  $g_1^{\tau_2}=L^{\tau_{02}}$, $g_2^{\rho_2}=R^{\rho_{02}}$;
  $g_1^\delta = L$, $g_2^\gamma=g_2^r$;
  $g_1^{\delta V}=L^V$, $(g_2^r)^{\varpi^T} = (g_2^\gamma)^{\varpi^T}$ and $(g_1^{\delta})^{\varpi^T*V} = (g_1^\gamma)^{\varpi^TV}$ and outputs:
  
$\texttt{param} = (m,n,b_1,b_2,c_1,c_2,d_1,d_2,p,e, \mathbb{G}_1, \mathbb{G}_2, \mathbb{G}_T,g_1,g_2,g_T)$.
  
  $EK_f=\{A,\omega^T,(g_1^{\tau_1^T}),(g_1^{\tau_2^T}),(g_1^{\eta^T}),(g_1^{\delta\cdot V})\}$

$VK_f=\{(g_1^{\tau_1^T}),(g_1^{\tau_2^T}),(g_2^{\rho_1}),(g_2^{\rho_2}),(g_1^{\eta^T}),(g_2^\mu),(g_1^{\delta\varpi^T\cdot V}),(g_2^{\gamma\varpi}),g_2^\gamma\}$.
  
  Thanks to the randomness and the decisional Diffie-Hellman assumption (DDH) in
  each group $G_i$, as well as \cite[Theorem~3]{Fiore:2012:PVD} for $\omega^T =
  (L^r)^{v^T}$, these public values are indistinguishable from randomly generated
  inputs.   
  Further, we have
  $\omega^T=g_1^{arv^T}=g_1^{ab(u_0^TA+t_0^T+v^T)+arv^T}=g_1^{abu_0^TA+abt_0^T+a(b+r)v^T}$.

  When adversary $\mathcal{A}$ calls the oracle $\mathcal{O}_{ProbGen}$ on input $x$, adversary $\mathcal{B}$ returns $\sigma_x = x$.
  Therefore, if $y=Ax$ and $\zeta=\omega^T\star{}x$, then
  the verification will pass:
  indeed the first two checks will ensure that 
  $s_1^T=g_1^{\tau_1^T}\star{}X$ and $s_2^T=g_1^{\tau_2^T}\star{}X$ when the third
  check ensures that $z^T=g_1^{\eta^T}\star{}Y$. 
  This shows that:
  $$H=\left(\prod{} e(z[i];g_1^{\mu[i]})\right)=e(g_1;g_2)^{abu_0^Ty},$$
  and that:
  $$D_j=\left(\prod e(s_j[i];g_2^{\rho_j[i]})\right)~\text{for}~j=1,2.$$
  Finally, the last check is that these two parts, as well as the last one, which
  is 
  $e(g_1^\delta;g_2^\gamma)^{v^T\cdot{}x}=e(g_1;g_2)^{a(b+r)v^T\cdot{}x}=
  e(Trace(g_1^{\delta\cdot{}VX});g_2^\gamma)$,
  are coherent with the definitions of $\omega$ and $\zeta$ above.
  Now, with a non-negligible probability $\epsilon$, adversary $\mathcal{A}$ can pass the check for another $y'\neq{}y$, by
  providing an adequate~$\zeta'$.
  First, $z^T$, $s_1^T$, $s_2^T$ and $C$ must be correct, as they are
  checked directly and independently by the Freivalds first four checks. 
  Second, we have that
  $e(\zeta';g_2)=e(g_1;g_2)^{abu_0^Ty'+abt_0^Tx}e(g_1^\delta;g_2^\gamma)^{v^T\cdot{}x}$
  and therefore, we must also have
  $e(\zeta(\zeta')^{-1};g_2)=e(g_1;g_2)^{abu_0^T(y-y')}$. 
  As $u_0$ is a secret unknown to adversary $\mathcal{A}$, for a random $y'$ the probability
  that $u_0^T(y-y')=0$ is bounded by $1/|\G_1|$ and thus negligible. 
  Thus adversary $\mathcal{B}$ can compute $c \equiv\left(u_0^T(y'-y)\right)^{-1} \mod |\G_1|$ and
  $(\zeta/\zeta')^c=g_1^{ab}$. Therefore it breaks the \mbox{co-CDH} assumption
  with non-negligible probability $\epsilon'\simeq \epsilon$. The only other
  possibility is that adversary $\mathcal{A}$ was able to recover $u_0^T$. But
  that would directly implies that it has an advantage in the \mbox{co-CDH}: $g_1^\eta=L^{\eta_0}$, $g_2^\mu=R^{\mu_0}$.
\end{proof}

In Table~\ref{tab:speedup}, we present more timings for the comparison between
our protocol and, to our knowledge and according to Table~\ref{tab:expes}, the
best previously known from~\cite{Elkhiyaoui:2016:ETPV}.  
The associated speed-ups supports our claim of a {\em Prover efficient} protocol
with a gain of two orders of magnitude. 
\begin{table}[htb]\centering
\caption{Speed-up of our novel Protocol over a 256-bit finite
  field on a i7 @3.4GHz.}\label{tab:speedup}
\begin{tabular}{|l||r|r|r|r|r|r|r|r|r|r|r|r|r|}
\hline
Size  & 100  & 200  & 500  & 1000  & 2000  & 3000  & 4000  \\
\hline
 \cite{Elkhiyaoui:2016:ETPV}   &  2.77s  &  10.93s  &  67.93s  &  271.03s  &  1079.71s  &  2430.05s  &  4329.46s  \\
 Fig.~\ref{fig:provprot}  &  0.17s  &  0.34s  &  0.98s  &  2.26s  &  5.37s  &  9.16s  &  13.76s  \\
\hline
 Speed-up  &  17  &  32  &  69  &  120  &  201  &  265  &  315  \\
\hline
\end{tabular}

\medskip

\begin{tabular}{|l||r|r|r|r|r|r|r|r|r|r|r|r|r|}
\hline
Size  & 5000  & 6000  & 7000  & 8000  & 9000  & 10000  \\
\hline
 \cite{Elkhiyaoui:2016:ETPV}   &  6790.15s  &  9780.24s  &  13309.61s  &  17416.38s  &  22002.51s  &  27175.12s  \\
 Fig.~\ref{fig:provprot}  &  18.55s  &  24.03s  &  29.93s  &  37.00s  &  44.00s  &  51.97s  \\
\hline
 Speed-up  &  366  &  407  &  445  &  471  &  500  &  523  \\
\hline
\end{tabular}
\end{table}




\end{document}